\documentclass[12pt,oneside,english]{amsart}
\usepackage[utf8]{inputenc}
\usepackage[a4paper]{geometry}
\usepackage{footnote}
\usepackage{color}
\usepackage{babel}
\usepackage{array}
\usepackage{float}
\usepackage{url}

\usepackage[foot]{amsaddr}

\usepackage{enumitem}
\usepackage{multirow}
\usepackage{amsmath}
\usepackage{amsthm}
\usepackage{graphicx}
\usepackage{setspace}
\usepackage[authoryear]{natbib}
\usepackage[unicode=true,pdfusetitle,
 bookmarks=true,bookmarksnumbered=false,bookmarksopen=false,
 breaklinks=false,pdfborder={0 0 1},backref=section,colorlinks=true]
 {hyperref}
\hypersetup{
 citecolor=cite-blue}

\makeatletter

  \theoremstyle{plain}
  \newtheorem{prop}{\protect\propositionname}

  \theoremstyle{plain}
  \newtheorem{cor}{\protect\corollaryname}
  \theoremstyle{remark}
  
  \theoremstyle{definition}
  \newtheorem{defn}{\protect\definitionname}
\theoremstyle{plain}
\newtheorem{thm}{\protect\theoremname}
  \theoremstyle{plain}
  \newtheorem*{lem*}{\protect\lemmaname}

\definecolor{cite-blue}{RGB}{0,0,204}
\date{}

\makeatother

  \providecommand{\definitionname}{Definition}
  \providecommand{\lemmaname}{Lemma}
  \providecommand{\propositionname}{Proposition}
  \providecommand{\remarkname}{Remark}
\providecommand{\corollaryname}{Corollary}
\providecommand{\theoremname}{Theorem}

\begin{document}

\title{Assignment Maximization}

\author{Mustafa O\u{g}uz Afacan}\thanks{\textbf{Mustafa O\u{g}uz Afacan}: Sabanc\i{} University, Faculty of Art and Social Sciences, Orhanli,
34956, Istanbul, Turkey. e-mail: mafacan@sabanciuniv.edu}

\author{Inácio Bó}
\thanks{\textbf{Inácio Bó}: University of York, Department of Economics and Related Studies. website: \protect\url{http://www.inaciobo.com}; e-mail:
inacio.lanaribo@york.ac.uk}

\author{Bertan Turhan}\thanks{\textbf{Bertan Turhan}: Iowa State University, Department of Economics, 260 Heady Hall, Ames, IA, 50011, USA. e-mail: bertan@iastate.edu} 

\date{}
\thanks{We thank Ahmet Alkan, Orhan Aygün, Mehmet Barlo, Umut Dur, Andrei Gomberg,
Isa Hafal\i r, Onur Kesten, Vikram Manjunath, Tridib Sharma, Tayfun
Sönmez, Alex Teytelboym, William Thomson, and Utku Ünver for helpful
comments. Afacan acknowledges the Marie Curie International Reintegration
Grant. Bó acknowledges financial support by the Deutsche Forschungsgemeinschaft
(KU 1971/3-1). }

\maketitle

\begin{abstract}
We evaluate the goal of maximizing the number of individuals matched to acceptable outcomes. We show that it implies incentive, fairness, and implementation impossibilities. Despite that, we present two classes of mechanisms that maximize assignments. The first are Pareto efficient, and undominated --- in terms of number of assignments --- in equilibrium. The second are fair for unassigned students and assign weakly more students than stable mechanisms in equilibrium. 
\end{abstract}


\emph{\hspace{-0.6cm}Keywords}: Market Design, Matching, Maximal
Matching, Fairness, Object Allocation, School Choice.

\section{Introduction\label{sec:Introduction}}

In this paper, we consider the economic problems faced by a market designer who wants
to produce student matchings (or object allocations) that are responsive to agents' preferences and leave the smallest number of them unmatched (that is, have maximum cardinality among individually rational matchings).\footnote{Even though our entire analysis translates naturally to most unit-demand discrete assignment problems, we will use the framing of school allocation throughout the paper.}

One of the main motivations for studying this problem is the fact that, in practice, market designers often make adaptations to standard procedures with the objective of preventing agents from being left unmatched. The real life use of allocation mechanisms in school choice procedures, for example, often consists of using a standard mechanism, such as the Gale-Shapley deferred acceptance \citep{gale1962college}, followed by some additional procedure to assign the students who were left unmatched into some school. These secondary steps or other ad-hoc methods for filling up the remaining seats, however, result in the loss of the properties of the mechanism that was used in the first place, such as fairness and strategy-proofness \citep{kestendur2019}. In this paper we start instead from the presumption that the market designer has the objective of leaving the minimum number of students unmatched. While this objective is not attainable via a strategy-proof mechanism, we propose mechanisms that produce maximum matchings and are efficient or satisfy a novel fairness criterion, when students are non-strategic. We also show that they satisfy desirable characteristics in equilibrium, and increase the cardinality of the matching as the proportion of truthful agents increases. 

\section{Related Literature\label{related}}

While algorithms for finding maximum matchings are well-known \citep{kuhn1955hungarian,Berge842}, the research on the incentives induced by the use of these procedures is limited, and typically rely on random mechanisms. One exception is \cite{afa-dur}, which follows-up to this paper and shows that no strategy-proof  and individually rational mechanism systematically matches more students than either of Boston, Gale-Shapley deferred acceptance, and serial dictatorship mechanisms. \cite{Krysta_2014} consider the problem of producing maximal matchings in a house allocation problem. They show that there is no mechanism that is deterministic, maximal, and strategy-proof, and provide instead a random mechanism that is strategy-proof and yields approximately-maximal outcomes. \cite{bogomolnaia2015size} evaluate the trade-off between maximality and envy-freeness, a notion of fairness that is stronger than the ones we consider in this paper. \cite{BogomolnaiaMoulin2004} consider the random assignment when agents have dichotomous preferences. When that is the case, Pareto efficiency is equivalent to maximality of the matching, and moreover, since agents are indifferent between all ``acceptable'' allocations, maximality doesn't result in incentive problems even in deterministic mechanisms.  \cite{noda}   studies the matching size achieved by strategy-proof mechanisms in a general model of matching with constraints.  

Finding the matching with maximum cardinality subject to some constraints is also a problem that is explored in the literature. \cite{IrvingManlove2010} consider the problem of finding stable matchings with maximum cardinality when priorities have ties, which is known to be an NP-hard problem, and present heuristics for finding them. \cite{aslagi20} also consider object assignment problems under distributional constraints. The authors show that variants of serial dictatorship and Probabilistic Serial \citep{bogomolnaia2001new} mechanisms assign at least as many agents as one can match under the constraints, while the violations of the constraints are relatively small.

Assignment maximization has been the primary objective in the organ exchange literature, as it means the maximum number of transplants.
This literature was initiated by the seminal work on kidney exchange  of \cite{kidney1}. In a subsequent study, in order to accommodate
several physical and geographical restrictions in operating transplants, \cite{roth2005pairwise} introduce the idea of pairwise kidney exchange where exchanges can only be made between two pairs.  They suggest implementing the priority-based maximal matching algorithm
from the combinatorial optimization literature \citep{book}. The first stages of both EAM and FAM are adaptations of the priority-based
maximal matching algorithm. Some other  studies on organ exchanges include  \cite{yilmaz}, \cite{tommy}, \cite{heo}, \cite{ergin2015lung}, 
 \cite{nic}, and \cite{ergin3}.

Refugee reassignment is another real-world application in which maximality  might be a primary design objective. \cite{andersson2016assigning} study the problem of finding housing for refugees
once they have been granted asylum. The authors propose an easy-to-implement
mechanism that finds an efficient stable maximum matching. They show
that such a matching guarantees that housing is efficiently provided
to a maximum number of refugees and that no unmatched refugee-landlord pair prefers each other.

Our ``fairness for unassigned students" is a weakening of the usual stability of \cite{gale1962college}, therefore, the current study is also related to the surging literature on the weakening of stability in different ways. Among others, \cite{dur}, \cite{afacan}, \cite{mor-ehler}, and \cite{troyan} are recent papers from that literature.

\section{Model\label{sec:Model}}

A $\mathbf{school}$ $\mathbf{choice}$ $\mathbf{problem}$ consists
of a finite set of students $I=\{i_{1},...,i_{n}\}$, a finite set of schools $S=\{s_{1},...,s_{m}\}$, a strict priority structure for schools $\succ=(\succ_{s})_{s\in S}$
where $\succ_{s}$ is a linear order over $I$, a capacity vector $q=(q_{s_{1}},...,q_{s_{m}})$, and a profile of strict preference of students $P=(P_{i})_{i\in I}$, where $P_{i}$ is student $i$'s preference relation over $S\cup\{\emptyset\}$ and $\emptyset$ denotes the option of being unassigned. We denote the set of all possible preferences for a student by $\mathcal{P}$. Let $R_{i}$ denote the \emph{at-least-as-good-as }preference relation associated with $P_{i}$, that is: $sR_{i}s^{'}$ $\Leftrightarrow$ $sP_{i}s^{'}$ or $s=s^{'}$. A school $s$ is \textbf{acceptable} to $i$ if $sP_{i}\emptyset$, and \textbf{unacceptable} otherwise. Let $Ac(P_i)=\{c\in S:\ c P_i \emptyset\}$.

In the rest of the paper, we consider the tuple $(I,S,\succ,q)$ as
the commonly known primitive of the problem and refer to it as the
\textbf{market}. We suppress all those from the problem notation and
simply write $P$ to denote the problem. A \textbf{matching} is a
function $\mu:I\rightarrow S\cup\{\emptyset\}$ such that for any
$s\in S$, $|\mu^{-1}(s)|\leq q_{s}$. A student $i$ is \textbf{assigned}
under $\mu$ if $\mu\left(i\right)\neq\emptyset$. For any $k\in I\cup S$,
we denote by $\mu_{k}$ the assignment of $k$. Let $|\mu|$ be the
total number of students assigned under $\mu$.

A matching $\mu$ is \textbf{individually rational} if, for any student
$i\in I$, $\mu_{i}R_{i}\emptyset$. A matching $\mu$ is \textbf{non-wasteful}
if for any school $s$ such that $sP_{i}\mu_{i}$ for some student
$i\in I$, $|\mu_{s}|=q_{s}$. A matching $\mu$ is \textbf{fair}
if there is no student-school pair $(i,s)$ such that $sP_{i}\mu_{i}$,
and for some student $j\in\mu_{s}$, $i\succ_{s}j$. A matching $\mu$
is \textbf{stable} if it is individually rational, non-wasteful, and
fair.

In the rest of the paper, we will consider only individually rational
matchings. Therefore, whenever we refer to a matching, unless explicitly
stated, we refer to an individually rational matching. Let $\mathcal{M}$
be the set of matchings.

A matching $\mu$ \textbf{dominates} another matching $\mu'$ if,
for any student $i\in S$, $\mu_{i}R_{i}\mu'_{i}$, and for some student
$j$, $\mu_{j}P_{j}\mu'_{j}$. A matching $\mu$ is \textbf{efficient}
if it is not dominated by any other matching. We say that
a matching $\mu$ \textbf{size-wise dominates} another matching $\mu'$
if $|\mu|>|\mu'|$. A matching $\mu$ is \textbf{maximal} if it is
not size-wise dominated.\footnote{Notice that the notions of size domination and maximality we use is in th set of \emph{agents} (or nodes) involved in a matching. In most of the literature in graph theory, the cardinality of a matching is measured in the set of edges of the graph that are part of the matching. While when considering the set of edges there is a difference between maximal and maximum cardinality matchings, in our setup these are equivalent: maximal matchings are always maximum.}

A \textbf{mechanism} $\psi$ is a function from $\mathcal{P}^{|I|}$
to $\mathcal{M}$. A Mechanism $\psi$ is \textbf{strategy-proof} if there
exist no problem $P$, and student $i$ with a false preference $P'_{i}$
such that $\psi_{i}(P'_{i},P_{-i})P_{i}\psi_{i}\left(P\right)$.

A mechanism $\psi$ \textbf{size-wise dominates}
another mechanism $\phi$ if, for any problem $P$, $\phi\left(P\right)$
does not size-wise dominate $\psi\left(P\right)$, while, for some
problem $P'$, $\psi\left(P'\right)$ size-wise dominates $\phi\left(P'\right)$.
A mechanism $\psi$ is maximal if it is not size-wise dominated by
any other mechanism.

We start our analysis by first observing that none among four well-known mechanisms commonly used and considered for the kind of allocation problems that we are considering --- deferred-acceptance ($DA$), top trading cycles ($TTC$), Boston ($BM$), and serial dictatorship ($SD$) --- is maximal.\footnote{For the description of these mechanisms, the reader could refer to \cite{abdulkadiroglu2003school}.}

\begin{prop}\label{pop: well-known mechanisms are not maxima}
None of $DA$, $TTC$, $BM$, and $SD$ is maximal.
\end{prop}
\begin{proof}
Let $I=\{i_1, i_2\}$ and $S=\{a,b\}$, each with unit quota. Let $P_{i_1}:\ a, b, \emptyset$ and $P_{i_2}: a, \emptyset$. The priorities are such that agent $i_1$ has the top priority at object $a$. Then, the $DA$, $TTC$, and $BM$ outcomes are the same. If we write $\mu$ for their outcome, then $\mu_{i_1}=a$ and $\mu_{i_2}=\emptyset$. Likewise, for $SD$, let us consider the ordering where agent $i_1$ comes first. Then, the $SD$ outcome is the same as $\mu$. This shows that none of these mechanisms is maximal because the matching $\mu'$ where $\mu'_{i_1}=b$ and $\mu'_{i_2}=a$ is individually rational and matches more agents than $\mu$. 
\end{proof}

Given the lack of maximality of the well-known mechanisms,  in the rest of the paper, we introduce two maximal mechanisms and study their properties.

\subsection{A Class of Efficient Maximal Mechanisms}

Given a problem
$P$ and an enumeration of the students in $I$ $\left(i_{1},..i_{n}\right)$,

\textbf{Step 0.} Let $\xi^{0}=\mathcal{M}$.

\textbf{Step 1.}

\textbf{Sub-step 1.1.} Define the set $\xi^{1}\subseteq\xi^{0}$ as
follows:
\begin{center}
\[
\xi^{1}=\left\{ \begin{array}{ll}
\{\mu\in\xi^{0}:\ \mu_{i_{1}}\neq\emptyset\} & \mbox{If \ensuremath{\exists\mu\in\xi^{0}\ }such that\ensuremath{\ \mu_{i_{1}}\neq\emptyset}}\\
\xi^{0} & \mbox{otherwise}
\end{array}\right.
\]
\par\end{center}

In general, for every $k\leq n$,

\textbf{Sub-step 1.k.} Define the set $\xi^{k}\subseteq\xi^{k-1}$
as follows:
\begin{center}
\[
\xi^{k}=\left\{ \begin{array}{ll}
\{\mu\in\xi^{k-1}:\ \mu_{i_{k}}\neq\emptyset\} & \mbox{If \ensuremath{\exists\mu\in\xi^{k-1}\ }such that\ensuremath{\ \mu_{i_{k}}\neq\emptyset}}\\
\xi^{k-1} & \mbox{otherwise}
\end{array}\right.
\]
\par\end{center}

Step $1$ ends with the selection of a matching $\mu\in\xi^{n}$.

\textbf{Step 2.}

In general:

\textbf{Sub-step 2.k.} Let $\tilde{\mu}$ be the matching obtained
in the previous step of the procedure. If $\tilde{\mu}$ does not admit an improving
chain or cycle then the algorithm ends with the final outcome of $\tilde{\mu}$.
Otherwise, pick such a chain or cycle, and obtain a new matching by
assigning each student in the chosen chain (cycle) to the school he
prefers in the chain (cycle), and move to the next sub-step.

\begin{thm}
\label{thm:EAM is maximal and  efficient} Every $EAM$ mechanism
is maximal and efficient.
\end{thm}
Notice, however, that while EAM mechanisms are maximal, they are not fair.

\subsection{A Class of Maximal and Fair for Unassigned Students Mechanisms}

We say that a matching $\mu$ is \textbf{fair for unassigned students}
if there is no student-school pair $(i,s)$ where $\mu_{i}=\emptyset$
and $i\succ_{s}j$ for some $j\in\mu_{s}$. A mechanism $\psi$ is
fair for unassigned students if, for any problem $P$, $\psi\left(P\right)$
is fair for unassigned students.

Below is a description of how each mechanism in this class works.
Given a problem $P$,

\textbf{Step 1.} Pick an $EAM$ mechanism $\psi$, and let $\psi\left(P\right)=\mu$.

\textbf{Step 2.}

In general,

\textbf{Sub-step 2.k.} Let $\tilde{\mu}$ be the matching obtained
in the previous step. If $\tilde{\mu}$ is fair for unassigned students,
the algorithm terminates with the outcome $\tilde{\mu}$. Otherwise,
pick a student-school pair $(i,s)$ such that $sP_{i}\emptyset$,
$\tilde{\mu}_{i}=\emptyset$, and $i\succ_{s}j$ for some $j\in\tilde{\mu}_{s}$.
Place student $i$ at school $s$, and let the lowest priority student
in $\tilde{\mu}_{s}$ be unassigned, while keeping everyone else's
assignment the same. Note that as in each sub-step the number of assigned
students is preserved, $\tilde{\mu}$ is maximal. Hence, we have $|\tilde{\mu}_{s}|=q_{s}$.
Let $\hat{\mu}$ be the obtained matching, and move to the next sub-step.

As, in every sub-step, a higher priority student is placed at a school
while a lower priority one is displaced from the school, and both
the students and schools are finite, the algorithm terminates in finitely
many rounds. The above procedure defines a class of mechanisms, each
of which is associated with different selections of the first stage
$EAM$ mechanism as well as the student-school pairs in the course
of Step $2$. We refer them as ``Fair Assignment
Maximizing'' ($FAM$) mechanisms.

\begin{thm}
\label{thm:FAMMaximalAndFair}Every $FAM$ mechanism is fair for unassigned
students and maximal.
\end{thm}
\begin{proof}
Let $\psi$ be a $FAM$ mechanism, and $\mu$ be the outcome of its
first step. As $\mu$ is the outcome of an $EAM$ mechanism, and in
Step $2$ of $\psi$, no student is assigned to one of his unacceptable
choices, $\psi$ is individually rational. Because $\mu$ is maximal
and the number of assigned students is preserved as $\left|\mu\right|$
in the course of Step $2$, $\psi$ is maximal. Moreover, as $\psi$
does not stop until no student-school pair violates fairness for unassigned
students, $\psi$ is fair for unassigned students as well.
\end{proof}

\section{Incentives and Equilibrium Analysis\label{sec:Incentives-and-equilibrium}}

In this section we show that the mechanisms in the classes
$EAM$ and $FAM$ have surprisingly regular properties in terms of
equilibrium outcomes. Consider the preference reporting game
induced by a mechanism $\psi$. At problem $P$, a preference submission
$P'=(P'_{i})_{i\in I}$ is a (Nash) \textbf{equilibrium} of $\psi$
if for every student $i$, $\psi_{i}\left(P'\right)R_{i}\psi_{i}(P''_{i},P'_{-i})$
for any $P''_{i}\in\mathcal{P}$. Let $\Omega$ be the set of mechanisms
that admit an equilibrium in any problem $P\in\mathcal{P}^{|I|}$.
In the rest of this section, we consider only the mechanisms in $\Omega$.

\begin{prop}
\label{prop:AMM-Equilibrium}Every $EAM$ and $FAM$ mechanism is
in $\Omega$. Moreover, for any problem, an $EAM$ mechanism has a
unique equilibrium outcome that is equivalent to the outcome of the
serial dictatorship where the student ordering is the same as that
used in that $EAM$ mechanism.
\end{prop}
Proposition \ref{prop:AMM-Equilibrium} shows, therefore, that equilibrium
outcomes of $EAM$ are not only Pareto efficient, but will match as
many students as a commonly used strategy-proof mechanism.

Our next question is how mechanisms compare, in terms of the number
of assignments, in equilibrium. For that, we define the concept of
\textbf{size-wise domination in equilibrium. }
\begin{defn}
For a given market $\left(I,S,\succ,q\right)$, a mechanism $\psi$
\textbf{size-wise dominates} another mechanism $\phi$ \textbf{in
equilibrium} if, for any problem $P$ and for every equilibria $P',P''$
under $\psi$ and $\phi$, respectively $\left|\psi\left(P'\right)\right|\geq\left|\phi\left(P''\right)\right|$,
and there exists a problem $P^{*}$ such that for every equilibria
$\hat{P},\tilde{P}$ under $\psi$ and $\phi$, respectively $|\psi(\hat{P})|>|\phi(\tilde{P})|$.
\end{defn}

\begin{thm}
\label{thm:EAMNotDominatedInEq}In any market $\left(I,S,\succ,q\right)$,
no $EAM$ mechanism is size-wise dominated by an individually rational
mechanism in equilibrium.
\end{thm}

While we do not have a similar result to above for the $FAM$ mechanisms,
we are able to compare the number of assigned students under the $FAM$
in equilibrium and the weakly dominant strategy equilibrium of the
$DA$, which is truth-telling.
\begin{thm}
\label{thm:FAMM-equilibriumNumber} Regarding the $FAM$ mechanisms:
\begin{itemize}
\item[(i)] For any problem $P$ and any stable matching for $P$ $\mu^{*}$,
for every equilibrium $P'$ of a $FAM$ mechanism $\psi$, $\left|\psi\left(P'\right)\right|\geq\left|\mu^{*}\right|$.
\item[(ii)] There exist a $FAM$ mechanism $\psi$, problem $P$, and an equilibrium
profile $P'$ of $\psi$ at $P$ such that $\left|\psi\left(P'\right)\right|>\left|\mu^{**}\right|$,
where $\mu^{**}$ is any stable matching for $P$.
\end{itemize}
\end{thm}

One may interpret the results in this section as an indication that there isn't much gain in using maximal mechanisms such as EAM and FAM, since when agents respond to their incentives, outcomes are similar to those produced by other non-maximal mechanisms. Below we show, however, that there are improvements in terms of the cardinality of the matching, as long as some fraction of the students are sincere.

\begin{prop}\label{prop:FractionOfTruthTelling}
For any maximal mechanism $\psi$, problem $P$, and student $i$ with false preferences $P'_i$ such that $\psi_i(P'_i,P_{-i}) P_i \psi_i(P)$, we have $|\psi(P)| \geq |\psi(P'_i,P_{-i})|$. Moreover, there exist a problem $\tilde{P}$ and student $i$ with false preferences $\bar{P}_i$ such that $\psi_i(\bar{P}_i,\tilde{P}_{-i}) \tilde{P}_i \psi_i(\tilde{P})$ and $|\psi(\tilde{P})| > |\psi(\bar{P}_i,\tilde{P}_i)|$. 
\end{prop}

In a preference-reporting game induced by a maximal mechanism where the only active players are strategic students in the sense that the rest is always sincere,  Proposition \ref{prop:FractionOfTruthTelling} leads to the following corollary.

\begin{cor}\label{cor:FractionOfTruthTelling}
Under any maximal mechanism, as the set of sincere students increases, in any problem, the number of students matched in equilibrium either stays the same or increases.
\end{cor}

\newpage 

\bibliographystyle{aer}
\bibliography{amm}

@article{IrvingManlove2010,
author = {Irving, Robert W. and Manlove, David F.},
title = {Finding Large Stable Matchings},
year = {2010},
issue_date = {2009},
publisher = {Association for Computing Machinery},
address = {New York, NY, USA},
volume = {14},
issn = {1084-6654},
url = {https://doi.org/10.1145/1498698.1537595},
doi = {10.1145/1498698.1537595},
journal = {ACM J. Exp. Algorithmics},
month = jan,
articleno = {2},
numpages = {30}
}

@article{bogomolnaia2001new,
  title={A new solution to the random assignment problem},
  author={Bogomolnaia, Anna and Moulin, Herv{\'e}},
  journal={Journal of Economic theory},
  volume={100},
  number={2},
  pages={295--328},
  year={2001},
  publisher={Elsevier}
}

@Article{aslagi20,
  author={Itai Ashlagi and Amin Saberi and Ali Shameli},
  title={{Assignment Mechanisms Under Distributional Constraints}},
  journal={Operations Research},
  year=2020,
  volume={68},
  number={2},
  pages={467-479},
  month={March},
  doi={10.1287/opre.2019.1887},
  url={https://ideas.repec.org/a/inm/oropre/v68y2020i2p467-479.html}
}

@article{kestendur2019,
	Author = {Dur, Umut and Kesten, Onur},
	Da = {2019/09/01},
	Doi = {10.1007/s00199-018-1133-9},
	Id = {Dur2019},
	Isbn = {1432-0479},
	Journal = {Economic Theory},
	Number = {2},
	Pages = {251--283},
	Title = {Sequential versus simultaneous assignment systems and two applications},
	Ty = {JOUR},
	Url = {https://doi.org/10.1007/s00199-018-1133-9},
	Volume = {68},
	Year = {2019}}

@article{roth1984b,
	Author = {Alvin E. Roth},
	Issn = {00223808, 1537534X},
	Journal = {Journal of Political Economy},
	Number = {6},
	Pages = {991-1016},
	Publisher = {University of Chicago Press},
	Title = {The Evolution of the Labor Market for Medical Interns and Residents: A Case Study in Game Theory},
	Url = {http://www.jstor.org/stable/1831989},
	Volume = {92},
	Year = {1984}}

@article{ergin2015lung,
	Author = {Ergin, Haluk I and Sonmez, Tayfun and {\"U}nver, M Utku},
	Journal = {Econometrica},
	Title = {Dual-Donor Organ Exchange},
	Year = {2017},
	pages={1645-1671},
	volume={85}
	
	}

@article{ergin3,
	Author = {Ergin, Haluk I and Sonmez, Tayfun and {\"U}nver, M Utku},

	Journal = {mimeo},
	Title = {Efficient and incentive compatible liver exchange},
	Year = {2018}
	
	
	}

@article{roth2005pairwise,
	Author = {Roth, Alvin E and S{\"o}nmez, Tayfun and {\"U}nver, M Utku},
	Journal = {Journal of Economic theory},
	Number = {2},
	Pages = {151--188},
	Publisher = {Elsevier},
	Title = {Pairwise kidney exchange},
	Volume = {125},
	Year = {2005}}

@article{kidney1,
	Author = {Roth, Alvin E and S{\"o}nmez, Tayfun and {\"U}nver, M Utku},
	
	Journal = {Quarterly Journal of Economics},

	Pages = {105-129},
	
	Title = {Kidney exchange},
	Volume = {119},
	Year = {2004}}

@article{abdulkadiroglu2003school,
	Author = {Abdulkadiroglu, Atila and S{\"o}nmez, Tayfun},
	Journal = {The American Economic Review},
	Number = {3},
	Pages = {729--747},
	Publisher = {American Economic Association},
	Title = {School choice: A mechanism design approach},
	Volume = {93},
	Year = {2003}}

@article{afa-dur,
	Author = {Afacan, Mustafa O and Dur, Umut M},
	Journal = {mimeo},

	Title = {Strategy-proof Size Improvement: Is it Possible?},
	
	Year = {2018}}

@article{gale1962college,
	Author = {Gale, David and Shapley, Lloyd S},
	Journal = {The American Mathematical Monthly},
	Number = {1},
	Pages = {9--15},
	Publisher = {JSTOR},
	Title = {College admissions and the stability of marriage},
	Volume = {69},
	Year = {1962}}

@article{andersson2016assigning,
	Author = {Andersson, Tommy and Ehlers, Lars},
	Journal = {forthcoming, Scandinavian Journal of Economics},

	Title = {Assigning Refugees to Landlords in Sweden: Efficient Stable Maximum Matchings},
	Year = {2018}}

@article{nic,
	Author = {Antonio Nicoló and Carmelo Rodríguez-Alvarez},
	Journal = {Games and Economic Behavior},
	Pages = {508-524},
	Title = {Age-based preferences in paired kidney exchanges},
	Volume = {102},
	Year = {2017}}

@book{book,
  title={Algorithms and Combinatorics},
  author={Korte, Bernhard and  Vygen, Jens},
   year={2011},
  journal={Combinatorial optimization: Theory and algorithms}
}

@article{yilmaz,
  title={How (not) to integrate blood
subtyping technology to kidney exchange},
  author={S{\"o}nmez, Tayfun and {\"U}nver, M Utku and Yilmaz, {\"O}zgur},
  journal={Journal of Economic Theory},
  volume={176},
  pages={193-231},
  year={2018},
  publisher={Elsevier}
}

@article{tommy,
  title={Pairwise kidney exchange over blood group
barrier},
  author={Tommy Andersson and Jörgen Kratz},
  journal={mimeo},
  year={2018}
}

@article{heo,
  title={Kidney exchange with immunosuppressants},
  author={Youngsub Chun and  Eun Jeong Heo and Sunghoon Hong},
  journal={mimeo},
  year={2018}
  
}

@article{dur,
  title={School choice with partial fairness},
  author={Umut Dur and Arda Gitmez and {\"O}zgur Yilmaz},
  journal={forthcoming, Theoretical Economics},
  year={2018},

}

@article{troyan,
  title={Essentially stable matchings},
  author={Peter Troyan and David Delacrétaz and Andrew Kloosterman},
  journal={mimeo},
 
  year={2018},

}

@article{mor-ehler,
  title={(Il)legal Assignments in School Choice},
  author={Thayer Morrill and Lars Ehlers},
  journal={mimeo},
  year={2018}

}

@article{afacan,
  title={Sticky Matching in School Choice},
  author={M.O Afacan and Z.H. Aliogullari and Mehmet Barlo},
  journal={Economic Theory},
  year={2017},
  pages={509-538},
  volume={64}

}

@article {Berge842,
	author = {Claude Berge},
	title = {Two theorems in graph theory},
	volume = {43},
	number = {9},
	pages = {842--844},
	year = {1957},
	doi = {10.1073/pnas.43.9.842},
	publisher = {National Academy of Sciences},
	issn = {0027-8424},
	URL = {http://www.pnas.org/content/43/9/842},
	eprint = {http://www.pnas.org/content/43/9/842.full.pdf},
	journal = {Proceedings of the National Academy of Sciences}
}

@article{kuhn1955hungarian,
  title={The Hungarian method for the assignment problem},
  author={Harold W Kuhn},
  journal={Naval research logistics quarterly},
  volume={2},
  number={1-2},
  pages={83--97},
  year={1955},
  publisher={Wiley Online Library}
}

@article{bogomolnaia2015size,
  title={Size versus fairness in the assignment problem},
  author={Anna Bogomolnaia and Herve Moulin},
  journal={Games and Economic Behavior},
  volume={90},
  pages={119--127},
  year={2015},
  publisher={Elsevier}
}

@article{noda,
  title={Large matching in large markets with flexible supply},
  author={Shunya Noda},
  journal={mimeo},
  year={2018}
  
}

@article{Krysta_2014,
   title={Size versus truthfulness in the house allocation problem},
   ISBN={9781450325653},
   url={http://dx.doi.org/10.1145/2600057.2602868},
   DOI={10.1145/2600057.2602868},
   journal={Proceedings of the fifteenth ACM conference on Economics and computation - EC  ’14},
   publisher={ACM Press},
   author={Piotr Krysta and David Manlove and Baharak Rastegari and Jinshan Zhang},
   year={2014}
}

@article{BogomolnaiaMoulin2004,
 ISSN = {00129682, 14680262},
 URL = {http://www.jstor.org/stable/3598855},
 author = {Anna Bogomolnaia and Herve Moulin},
 journal = {Econometrica},
 number = {1},
 pages = {257--279},
 publisher = {[Wiley, Econometric Society]},
 title = {Random Matching under Dichotomous Preferences},
 volume = {72},
 year = {2004}
}

\newpage

\section*{Appendix}

\subsection*{Proofs}

\subsubsection*{Theorem \ref{thm:EAM is maximal and  efficient}}

We will use the following Lemma:
\begin{lem*}
A maximal matching $\mu$ is efficient if and only if it does not
admit an improving chain or cycle.
\end{lem*}
\begin{proof}
\textbf{``Only If'' Part.} Let $\mu$ be an efficient matching. If
it admits an improving chain $\{i_1,..i_n, c_1, .., c_{n+1}\}$, then we can define a  new matching by assigning each agent $i_k$ to $c_{k+1}$ while
keeping the assignments of the others the same. By the improving chain definition, that new matching dominates $\mu$, contradicting our starting supposition that $\mu$ is efficient. The same argument shows for the case of improving cycle.

\textbf{``If'' Part.} Let $\mu$ be a maximal matching that does not admit improving chains or cycles. Assume for a contradiction that there exists a matching
$\mu'$ that dominates $\mu$.

Let $W=\{i\in I:\ \mu'_{i}P_{i}\mu_{i}\}$. By the supposition, $W\neq\emptyset$.
Note that for any student $i$ with $\mu_{i}\neq\emptyset$, we have
$\mu'_{i}\neq\emptyset$. This, along with the maximality of $\mu$,
implies that $|\mu'|=|\mu|$. Hence, for any student $i$ with $\mu_{i}=\emptyset$,
$\mu'_{i}=\emptyset$.

Enumerate the students in $W=\{i_{1},..,i_{n}\}$ and write
$\mu'_{i_{k}}=c_{k}$ for any $k=1,..,n$. If $|\mu_{c_{k}}|<q_{c_{k}}$
for some $k$, then the pair $\{i_{k},c_{k}\}$ constitutes an
improving chain, a contradiction.

Suppose that $|\mu_{c_{k}}|=q_{c_{k}}$ for any $k=1,..,n$.
As $c_{1}$ does not have excess capacity at $\mu$, and $\mu'_{i_{1}}=c_{1}$,
we have another student in $W$, say $i_{2}$, such that $\mu_{i_{2}}=c_{1}$.
Then, consider student $i_{2}$, and as $c_{2}$ does not have excess
capacity at $\mu$ and $\mu'_{i_{2}}=c_{2}$, we have another student
in $W$, say $i_{3}$, such that $\mu_{i_{3}}=c_{2}$. If we continue
to apply the same arguments to the other students in $W$, as $W$
is finite, we would eventually obtain an improving cycle, a contradiction.
\end{proof}

Let now $\psi$ be an
$EAM$ mechanism, and $\mu$ and $\mu'$ be its first stage and final
outcome, respectively. As students are not assigned to one of their
unacceptable schools in Step $1$ of $\psi$, $\mu$
is individually rational.

Assume for a contradiction that $\mu$ is not maximal and there exists 
$\mu''\neq \mu$ such that $|\mu''| > |\mu|$. Let $\{i_1,..,i_n\}$ be the agent-enumeration that is used under $\psi$.

As $|\mu''| > |\mu|$, there exists some agent $i_k\in I$ such that $\mu''_{i_k}\neq \emptyset$ and $\mu_{i_k}=\emptyset$. Let $i_{k'}$ be the first agent according to the above enumeration such that $\mu''_{i_{k'}}\neq \emptyset$ and $\mu_{i_{k'}}=\emptyset$. This means that for each $k < k'$, either $\mu_{i_k}\neq \emptyset$ or $\mu_{i_k}=\emptyset$ and $\mu''_{i_k}=\emptyset$.   Let $B(\mu,k')=\{i\in N:\ \mu_{i_{k}}\neq \emptyset$ for any $k < k'\}$. That is, it is set of agents who come before agent $i_{k'}$ in the above enumeration and are assigned under matching $\mu$.

Now consider agent $i_{k'}$. By the definition of $\psi$, $\mu_{i_{k'}}=\emptyset$ because it is not possible to match agent $i_{k'}$ to some of his acceptable objects while keeping all the agents in $B(\mu, k')$ assigned to one of their acceptable objects. This means that in order for agent $i_{k'}$ to receive one of his acceptable objects, one of the assigned agents under $\mu$ from $B(\mu, k')$ has to be unassigned. This arguments holds for each other agent who is assigned under $\mu''$, but not under $\mu$. This implies that $\mu$ is maximal.

In Step $2$ of $\psi$, new matchings are obtained by implementing
improving chains and cycles (if any). By their definitions no student receives a worse school than his assignment
$\mu$. This, along with the individual rationality of $\mu$, implies
that $\mu'$ is maximal. The efficiency of $\mu'$ directly comes
from the Lemma above.

\subsubsection*{Proposition \ref{prop:AMM-Equilibrium}}

Let $\psi$ be an $EAM$ mechanism. The first student
in Step $0$ of the $EAM$ obtains his top choice
by reporting it as the only acceptable choice. By the same reasoning, the
second student can obtain his top choice among the remaining schools with seats after considering the first student's assignment by reporting that school as his only
acceptable choice. Once we repeat the same arguments for every other student,
we not only find an equilibrium of $\psi$, but also conclude that
it is the unique equilibrium outcome, which coincides with the outcome
of serial dictatorship with the ordering being the same as that in Step
$0$ of $\psi$.

Let $\phi$ be a $FAM$ mechanism. Let $\mu$ be a stable matching
at $P$. Consider the preferences submission $P'$ under which for
any student $i$, the only acceptable school is $\mu_{i}$. Any unassigned
student at $\mu$ reports no school acceptable at $P'$. It is easy
to verify that $\phi\left(P'\right)=\mu$.

Next, we claim that $P'$ is an equilibrium submission under $\phi$.
Suppose for a contradiction that there exist a student $i$ and $P''_{i}$
such that $\phi_{i}(P''_{i},P'_{-i})P_{i}\phi_{i}\left(P'\right)$.
For ease of writing, let $\phi_{i}(P''_{i},P'_{-i})=s$ and $\phi_{i}\left(P'\right)=s'$.
As $\mu$ is stable, $|\mu_{s}|=q_{s}$. This, along with the definition
of $P'$ and $\phi_{i}(P''_{i},P'_{-i})=s$, implies that there exists
a student $j\neq i$ such that $\mu_{j}=s$ and $\phi_{j}(P''_{i},P'_{-i})=\emptyset$.
Moreover, from the stability of $\mu$, we also have $j\succ_{s}i$.
These altogether contradict the fairness for unassigned students of
$\phi$, showing that $P'$ is equilibrium of $\phi$.

\subsubsection*{Theorem \ref{thm:EAMNotDominatedInEq}}

We will use the following.
\begin{lem*}
\label{lemma} Let $\psi$ be an $EAM$ and $\phi$ be an individually
rational mechanism. In any market $(I,S,\succ,q)$ and problem $P$,
if $\left|\psi\left(P'\right)\right|<|\phi\left(P''\right)|$ where
$P'$ and $P''$ are equilibria under $\psi$ and $\phi$, respectively,
then there exists a student $i$ such that $\psi_{i}\left(P'\right)P_{i}\phi_{i}\left(P''\right)P_{i}\emptyset$.
\end{lem*}
\begin{proof}
In a market $(I,S,\succ,q)$ and problem $P$, let $\left|\psi\left(P'\right)\right|<|\phi\left(P''\right)|$
where $P'$ and $P''$ are equilibria under $\psi$ and $\phi$, respectively.
This implies that for some school $s$, $\left|\psi_{s}\left(P'\right)\right|<\left|\phi_{s}\left(P''\right)\right|\leq q_{s}$.
Hence, let $i\in\phi_{s}\left(P''\right)\setminus\psi_{s}\left(P'\right)$.
By the individual rationality of $\phi$ and $P''$ being equilibrium
under $\phi$, we have $sP_{i}\emptyset$, where $\phi_{i}\left(P''\right)=s$.
As the unique equilibrium outcome of $\psi$ coincides with the (truthtelling)
outcome of a $SD$ mechanism (Proposition $5$), we have $\psi\left(P'\right)=SD\left(P\right)$.
Hence, school $s$ has an excess capacity under $SD\left(P\right)$.
Moreover, from above, $\psi_{i}\left(P'\right)=SD_{i}\left(P\right)\neq s$.
Hence, by the non-wastefulness of $SD$, $i$ must be matched to a
school strictly better than $s$ and therefore $\psi_{i}\left(P'\right)=SD_{i}\left(P\right)P_{i}\phi_{i}\left(P''\right)P_{i}\emptyset$,
which finishes the proof.
\end{proof}
Let now $(I,S,\succ,q)$ be a market and $\psi$ be an $EAM$ mechanism.
Assume for a contradiction that an individually rational mechanism
$\phi$ size-wise dominates $\psi$ in equilibrium. This in particular
implies that for some problem $P$, $\left|\psi\left(P'\right)\right|<|\phi\left(P''\right)|$
for every equilibria $P'$ and $P''$ under $\psi$ and $\phi$, respectively.
In what follows, we will fix one such pair $P',P''$. We prove the
result in two steps.

\textbf{Step 1.} By the Lemma above, there exists a student $i$ such
that $\psi_{i}\left(P'\right)P_{i}\phi_{i}\left(P''\right)P_{i}\emptyset$.
Let $\bar{P}_{i}$ be the preference relation that keeps the relative
rankings of the schools the same as under $P_{i}$, while reporting
any school that is worse than $\psi_{i}\left(P'\right)$ as unacceptable.
In other words, $\bar{P}_{i}$ truncates $P_{i}$ below $\psi_{i}\left(P'\right)$.
Let us write $\bar{P}=(\bar{P}_{i},P_{-i})$. Recall that the unique
equilibrium outcome of $\psi$ always coincides with the truthtelling
outcome of a $SD$ mechanism (Proposition $5$). Moreover, by the
construction of $\bar{P}$, $SD\left(P\right)=SD(\bar{P})$. This
in turn implies that $\psi\left(P'\right)=\psi\left(\bar{P}'\right)$
for every equilibrium $\bar{P}'$ under $\psi$ in problem $\bar{P}$.

We next consider problem $\bar{P}$. If there exists no student $j$
such that $\psi_{j}\left(\bar{P}'\right)\bar{P}_{j}\phi_{j}\left(\bar{P}''\right)\bar{P}_{j}\emptyset$
for some equilibria $\bar{P}'$ and $\bar{P}''$ under $\psi$ and
$\phi$, respectively, then we move to Step $2$. Otherwise, we pick
such student $j$. Note that because of the definition of $\bar{P}_{i}$
states that any outcome below $\psi_{i}\left(\bar{P}'\right)$ is
unacceptable for $i$ and $\phi$ is individually rational, $\psi_{j}\left(\bar{P}'\right)\bar{P}_{j}\phi_{j}\left(\bar{P}''\right)\bar{P}_{j}\emptyset$
cannot hold for $j=i$, therefore $j\neq i$. Then, as the same as
above, let $\bar{P}_{j}$ be the preference list that truncates $P_{j}$
below $\psi_{j}\left(\bar{P}'\right)$. Let us write $\tilde{P}=(\bar{P}_{i},\bar{P}_{j},P_{-\{i,j\}})$.
By the same reason as above, $\psi\left(P'\right)=\psi\left(\tilde{P}'\right)$
for any equilibrium $\tilde{P}'$ under $\psi$ in problem $\tilde{P}$.

We next consider problem $\tilde{P}$. If there exists no student
$k$ such that $\psi_{k}\left(\tilde{P}'\right)\tilde{P}_{k}\phi_{k}\left(\tilde{P}''\right)\tilde{P}_{k}\emptyset$
for some equilibria $\tilde{P}'$ and $\tilde{P}''$ under $\psi$
and $\phi$, respectively, then we move to Step $2$. Otherwise, we
pick such a student $k$. By the same reason as above, student $k$
is different than both $i$ and $j$. Then, we follow the same arguments
above and obtain a new preference profile. In each iteration, we have
to consider a different student. But then, since there are finitely
many students, this case cannot hold forever. Hence, we eventually
obtain a problem, say $\hat{P}$, in which there exists no student
$h$ such that $\psi_{h}(\hat{P}')\hat{P}_{h}\phi_{h}\left(\hat{P}''\right)\hat{P}_{h}\emptyset$
for some equilibria $\hat{P}'$ and $\hat{P}''$ under $\psi$ and
$\phi$, respectively, and move to Step $2$. We also have $\psi\left(P'\right)=\psi(\hat{P}')$
for any equilibrium $\hat{P}'$ under $\psi$ in problem $\hat{P}$.

\textbf{Step $2$.} By the Lemma above, in problem $\hat{P}$, we
have $\left|\psi\left(\hat{P}'\right)\right|\geq\left|\phi\left(\hat{P}''\right)\right|$
for any equilibria $\hat{P}'$ and $\hat{P}''$ under $\psi$ and
$\phi$, respectively. If it holds strictly for some equilibria, then
we reach a contradiction. Suppose $\left|\psi(\hat{P}')\right|=\left|\phi\left(\hat{P}''\right)\right|$
for any equilibria $\hat{P}'$ and $\hat{P}''$.

We now claim that $\hat{P}''$ is an equilibrium under $\phi$ in
problem $P$. Suppose it is not, and let student $k$ have a profitable
deviation, say $\ddot{P}_{k}$, from $\hat{P}''_{k}$. This means
that $\phi_{k}\left(\ddot{P}_{k},\hat{P}''_{-k}\right)P_{k}\phi_{k}\left(\hat{P}''\right)$.
But then, by construction above, $\hat{P}_{k}$ preserves the relative
rankings under $P_{k}$. This implies that $\phi_{k}\left(\ddot{P}_{k},\hat{P}''_{-k}\right)\hat{P}_{k}\phi_{k}\left(\hat{P}''\right)$,
contradicting $\hat{P}''$ being an equilibrium under $\phi$ in problem
$\hat{P}$.

Recall that $\psi\left(P'\right)=\psi(\hat{P}')$. Hence, this, along
with $\left|\psi(\hat{P}')\right|=\left|\phi\left(\hat{P}''\right)\right|$
and our above finding, implies that in problem $P$, $\left|\psi\left(P'\right)\right|=\left|\phi\left(\hat{P}''\right)\right|$
where $P'$ and $\hat{P}''$ are equilibria under $\psi$ and $\phi$,
respectively. Therefore, we constructed an equilibrium pair for problem
$P$ where $\psi$ matches as many students as $\phi$, contradicting
our assumption that this does not hold in problem $P$.

\subsubsection*{Theorem \ref{thm:FAMM-equilibriumNumber}}

$(i)$. First, by the rural hospital theorem \citep{roth1984b}, the
number of assignments in any stable matching is the same as that of
DA. Let $\psi$ be a $FAM$ mechanism. Assume for a contradiction
that there exist a problem $P$ and an equilibrium profile $P'$ under
$\psi$ such that $\left|\psi\left(P'\right)\right|<|DA\left(P\right)|$.
For ease of writing, let $DA\left(P\right)=\mu$ and $\psi\left(P'\right)=\mu'$.

We now claim that for some student $i$, $\mu_{i}=s$ for some school
$s$ whereas $\mu'_{i}=\emptyset$ and, moreover, $|\mu'_{s}|<q_{s}$.
To prove this claim, let us define $W=\{i\in I:\ \mu_{i}=s$ and $\mu'_{i}=\emptyset\}$.
By our supposition that $|DA\left(P\right)|>\left|\psi\left(P'\right)\right|$,
we have $W\neq\emptyset$. Suppose that for each $i\in W$ with $\mu_{i}=s$,
$|\mu'_{s}|=q_{s}$. But then this implies that $|\mu'|\geq|\mu|$,
contradicting our initial supposition, which finishes the proof of
the claim.

Let $i\in I$ such that $\mu_{i}=s$, $\mu'_{i}=\emptyset$, and $|\mu'_{s}|<q_{s}$.
Now, consider the following preferences $P''$:
\begin{center}
\[
P''_{k}=\left\{ \begin{array}{ll}
P'_{k} & \mbox{If \ensuremath{k\neq i}}\\
s,\emptyset & \mbox{If \ensuremath{k=i}}
\end{array}\right.
\]
\par\end{center}

First, observe that there exists a (individually rational) matching
at $P''$ that assigns $|\mu'|+1$ many students (to see this, keep
the assignment of everyone except student $i$ the same as at $\mu'$,
and place student $i$ at school $s$). Therefore, due to the maximality
of $\psi$, we have $\left|\psi\left(P''\right)\right|\geq|\mu'|+1$.
If student $i$ is assigned to school $s$ at $\psi\left(P''\right)$
then this contradicts $P'$ being equilibrium under $\psi$. Hence,
$\psi_{i}\left(P''\right)=\emptyset$. But then, by the definition
of $P''$, $\psi\left(P''\right)$ is individually rational at $P'$.
This, along with the maximality of $\psi$, implies that $\left|\psi\left(P'\right)\right|\geq\left|\psi\left(P''\right)\right|$,
contradicting our previous finding that $\left|\psi\left(P''\right)\right|\geq\left|\psi\left(P'\right)\right|+1$,
which finishes the proof of the first part.

$(ii)$. Let us consider $I=\{i,j,k,h\}$ and $S=\{a,b,c\}$, each
with unit capacity. The preferences and the priorities are given below.
\begin{center}
$P_{i}:\ a,b,\emptyset$; $P_{j}:\ c,a,\emptyset$; $P_{k}:\ c,a,\emptyset$;
$P_{h}:\ c,\emptyset$.
\par\end{center}

\begin{center}
$\succ_{a}:\ k,i,j,h$; $\succ_{b}:\ k,h,j,i$; $\succ_{c}:\ k,h,i,j$.
\par\end{center}

Let $\psi$ be a $FAM$ mechanism with the student ordering $k,j,i,h$.
Mechanism $\psi$ is such that it produces matching $\mu$ at $P$
where $\mu_{i}=b$, $\mu_{j}=a$, $\mu_{k}=c$, and $\mu_{h}=\emptyset$.
For any $P'_{i}\in\mathcal{P}$ with $bP'_{i}\emptyset$, let $\psi(P'_{i},P_{-i})=\mu'$
where $\mu'_{i}=b$, $\mu'_{j}=\emptyset$, $\mu'_{k}=a$, and $\mu'_{h}=c$.
Moreover, for any $P'_{i}\in\mathcal{P}$ with $\emptyset P'_{i}b$,
$\psi(P'_{i},P_{-i})=\mu''$ where $\mu''_{i}=\emptyset$, $\mu''_{j}=\emptyset$,
$\mu''_{k}=a$, and $\mu''_{h}=c$. And, for any $P'_{h}\in\mathcal{P}$,
let $\psi(P_{-h},P'_{h})=\mu$.

Note that student $j$ can never get school $c$ under $\psi$ by
misreporting because otherwise student $h$ would be unassigned, and
he has higher priority at school $c$. It is immediate to see that
the above matchings can be obtained in the course of $FAM$ through
particular selection. All of these show that under $\psi$, truth-telling
is an equilibrium at $P$, and $|\psi\left(P\right)|=3$. On the other
hand, $DA\left(P\right)$ is such that $DA_{i}\left(P\right)=a$,
$DA_{k}\left(P\right)=c$, and $DA_{h}\left(P\right)=DA_{j}\left(P\right)=\emptyset$.
Hence, $|\psi\left(P\right)|>|DA\left(P\right)|$, finishing the proof
of the second part.

\subsubsection*{Proposition \ref{prop:FractionOfTruthTelling}}

Let $P'=(P'_i,P_{-i})$, $\psi(P)=\mu$, and $\psi(P'_i,P_{-i})=\mu'$. Assume that $|\mu'| > |\mu|$. By our supposition, $\mu'_i P_i \mu_i$. This, along with the fact that $P_j=P'_j$ for each $j\neq i$, $\mu'$ is individually rational in problem $P$. But then, $|\mu'| > |\mu|$ contradicts the fact  that $\mu$ is maximal in problem $P$.

Consider a problem where $\{i,j\}\subseteq N$, $\{a,b\}\subseteq S$, each with unit capacity. Let $P_i:\ a, \emptyset$, $P_j:\ a, \emptyset$, and each other student (if any) finds every school unacceptable. Without loss of generality, assume that the outcome of $\psi$ in that problem, say $\mu$,  is such that $\mu_i=a$, and each other student is unassigned. 

Consider a problem where $P'_i:\ a, b, \emptyset$, while each other student's preferences are the same as above.   Under the true preferences, $\psi$ produces $\mu'$ where $\mu'_i=b$, $\mu'_j=a$, and each other student is unassigned. However, student $i$ can misreport his preferences by submitting $P_i$ above as, under this false profile, $\psi$ produces matching $\mu$ above.  Finally, note that $|\mu'| > |\mu|$, finishing the proof. 
\newpage

\end{document}